\documentclass[a4paper,11pt]{article}

\usepackage{amsfonts}
\usepackage{amsthm}
\usepackage{amsmath}
\usepackage{authblk}
\usepackage{bm}
\usepackage{color}
\usepackage{graphicx}
\usepackage{enumitem}
\usepackage{lscape}
\usepackage{mathtools}
\usepackage{url}

\newcommand{\FF}{\mathbb{F}}
\newcommand{\ZZ}{\mathbb{Z}}
\newcommand{\0}{\mathbf{0}}
\newcommand{\1}{\mathbf{1}}

\DeclareMathOperator{\wt}{wt}
\DeclareMathOperator{\Hull}{Hull}
\DeclareMathOperator{\rank}{rank}
\DeclarePairedDelimiter\floor{\lfloor}{\rfloor}

\theoremstyle{plain}
\newtheorem{thm}{Theorem}[section]
\newtheorem{lem}[thm]{Lemma}
\newtheorem{cor}[thm]{Corollary}
\newtheorem{prp}[thm]{Proposition}
\theoremstyle{definition}

\theoremstyle{remark}
\newtheorem{rem}[thm]{Remark}


\title{Construction for both self-dual codes and LCD codes}
\author{Keita Ishizuka\thanks{Corresponding author. Research Center for Pure and Applied Mathematics Graduate School of Information Sciences, Tohoku University, Sendai 980–8579, Japan. email: \texttt{keita.ishizuka.p5@dc.tohoku.ac.jp}}, Ken Saito\thanks{Research Center for Pure and Applied Mathematics Graduate School of Information Sciences, Tohoku University, Sendai 980–8579, Japan. email: \texttt{kensaito@ims.is.tohoku.ac.jp}}}
\date{}

\begin{document}
\maketitle

\begin{abstract}
    From a given $[n,k]$ code $C$, we give a method for constructing many $[n,k]$ codes $C'$ such that the hull dimensions of $C$ and $C'$ are identical.
    This method can be applied to constructions of both self-dual codes and linear complementary dual codes (LCD codes for short).
    Using the method, we construct $661$ new inequivalent extremal doubly even $[56,28,12]$ codes.
    Furthermore, constructing LCD codes by the method, we improve some of the previously known lower bounds on the largest minimum weights of binary LCD codes of length $n=26,28 \le n \le 40$.
    \\
    \textbf{2010 AMS Classification.} 94B05
\end{abstract}

\section{Introduction} \label{sec:introduction}
    Let $\FF_q$ be the finite field of order $q$, where $q$ is a prime power.
    An $[n,k]$ code $C$ over $\FF_q$ is said to be a self-dual code if $C = C^\perp$, where $C^\perp$ denotes the dual code of $C$.
    A code is said to be doubly even if all codewords have weights divisible by four.
    Mallows and Sloane~\cite{mallows1973anupper} proved that the minimum weight $d$ of a binary doubly even self-dual code of length $n$ is upper bounded by $d \le 4 \floor{n/24} + 4$.
    A binary doubly even self-dual code meeting the bound is called extremal.
    An $[n,k]$ code $C$ over $\FF_q$ is said to be an LCD code if $C \cap C^{\perp} = \{ \0_n \}$, where $\0_n$ denotes the zero vector of length $n$.
    The concept of LCD codes was invented by Massey~\cite{massey1992linear}.
    A binary LCD $[n,k]$ code is said to be optimal if it has the largest minimum weight among all binary LCD $[n,k]$ codes.

    Although the definitions say that self-dual codes and LCD codes are quite different classes of codes, codes of both classes have similar properties.
    For example, it is known that both self-dual codes and LCD codes are characterized by their generator matrices.
    Furthermore, self-dual codes are codes with maximal hull dimension and LCD codes are codes with minimal hull dimension, where the hull of a code $C$ is defined as $\Hull(C) = C \cap C^\perp$.
    Recently Harada~\cite{harada2021construction} gave a method for constructing LCD codes modifying known methods for self-dual codes in~\cite[Theorem 2.2]{harada1996existence} and~\cite[Theorem 2.2]{harada1995new} and constructed $15$ optimal binary LCD $[n,k]$ codes.

    In this paper, we give a method for constructing many $[n,k]$ code $C'$ from a given $[n,k]$ code $C$ such that $\dim(\Hull(C)) = \dim(\Hull(C'))$.
    This method can be applied to constructions of both self-dual codes and LCD codes.
    It is shown that the method is a generalized version of~\textup{\cite[Theorem 2.2]{harada1996existence}},~\textup{\cite[Theorem 3.3]{harada2021construction}} and~\textup{\cite[Theorem 2.2]{harada1995new}}.
    As an application, we construct $661$ new inequivalent extremal doubly even $[56,28,12]$ codes.
    Furthermore, constructing LCD codes by the method, we improve some of the previously known lower bounds on the largest minimum weights of binary LCD codes of length $n=26,28 \le n \le 40$.

    This paper is organized as follows:
    In Section~\ref{sec:preliminaries}, we recall some basic results on self-dual codes, LCD codes and hulls of codes.
    In Section~\ref{sec:constructMethod}, we provide the construction method.
    Furthermore, in Section~\ref{sec:basicProperties}, we state basic properties of the construction method.
    In Section~\ref{sec:SD}, we construct $661$ new inequivalent extremal doubly even $[56,28,12]$ from six bordered double circulant doubly even $[56,28,12]$ codes.
    In Section~\ref{sec:LCD}, we improve some of the largest minimum weights among all binary LCD $[n,k]$ codes with length $n=26,28 \le n \le 40$, which were recently studied by Bouyuklieva~\cite{bouyuklieva2020optimal} and Harada~\cite{harada2021construction}.
    All computations in this paper were performed in MAGMA~\cite{bosma1997magma}.

\section{Preliminaries} \label{sec:preliminaries}
Let $\FF_q$ be the finite field of order $q$, where $q$ is a prime power and let $\FF_q^n$ be the vector space of all $n$-tuples over $\FF_q$.
A $k$-dimensional subspace of $\FF_q^n$ is said to be an $[n,k]$ code over $\FF_q$. 
Especially, codes over $\FF_2$ are said to be binary codes.
Let $C$ be an $[n,k]$ code over $\FF_q$. 
The parameters $n$, $k$ are said to be the length, the dimension of $C$ respectively.
A vector in $C$ is said to be a codeword.
The weight of $x=(x_1,x_2,\dots,x_n) \in \FF_q^n$ is defined as $\wt(x) = \# \{ i \mid x_i \ne 0 \}$.
The minimum weight of $C$ is defined as $\wt(C) = \min \{ \wt(x) \mid x \in C, x \neq \0_n \}$.
If the minimum weight of $C$ equals to $d$, then $C$ is said to be an $[n, k, d]$ code over $\FF_q$.
A code $C$ is said to be an even code if all codewords have even weights.
Also, a code is said to be a doubly even code if all codewords have weights divisible by four.
Two $[n,k]$ codes $C_1,C_2$ over $\FF_q$ are equivalent if there exists a monomial matrix $M$ such that $C_2 = \{cM \mid c \in C_1 \}$.
The equivalence of two codes $C_1, C_2$ is denoted by $C_1 \simeq C_2$.
A generator matrix of a code $C$ is any matrix whose rows form a basis of $C$. 

The dual code $C^\perp$ of an $[n,k]$ code $C$ over $\FF_q$ is defined as $C^{\perp} = \{ x \in \FF_q^n \mid (x, y) = 0 \text{ for all } y \in C \}$, where $(x,y)$ is the standard inner product.
If $C \subset C^\perp$, then $C$ is said to be a self-orthogonal code.
If $C = C^\perp$, then $C$ is said to be a self-dual code.
A binary self-dual code $C$ is doubly even if and only if $n \equiv 0 \pmod 8$, where $n$ denotes the length of $C$.
Mallows and Sloane~\cite{mallows1973anupper} proved that the minimum weight $d$ of a binary doubly even self-dual code of length $n$ is upper bounded by $d \le 4 \floor{n/24} + 4$.
A binary doubly even self-dual code meeting the bound is called extremal.
\begin{lem}[{\cite[Theorem 1.4.8]{huffman2010fundamentals}}] \label{F2:lem:doublyevencode}
    Let $C$ be a binary code.
    Then the following holds:
    \begin{enumerate}
        \item If $C$ is a self-orthogonal code and has a generator matrix each of whose rows has weight divisible by four, then $C$ is doubly even.
        \item If $C$ is doubly even, then $C$ is a self-orthogonal code.
    \end{enumerate}
\end{lem}

A pure double circulant code has a generator matrix of the form $\begin{pmatrix} I_k & R \\ \end{pmatrix}$ and a bordered double circulant code has a generator matrix of the form
\begin{equation} \label{eq:doubleCirculant}
    \begin{pmatrix}
        \begin{matrix}
        & & & \\
        & I_k & & \\
        & & &
        \end{matrix}
        \begin{matrix}
            0 & 1 & \ldots & 1 \\
            1 & & & \\
            \vdots &  & R &\\
            1 & & &
        \end{matrix}
    \end{pmatrix},
\end{equation}
where $I_k$ denotes the identity matrix of order $k$ and $R$ is a circulant matrix.
These two families of codes are collectively called double circulant codes.
Harada, Gulliver and Kaneta~\cite{harada1998classification} showed that there exist exactly nine inequivalent extremal double circulant doubly even $[56,28,12]$ codes and all of them are bordered double circulant codes.
In Section~\ref{sec:SD}, we construct extremal doubly even self-dual $[56,28,12]$ codes from six inequivalent extremal double circulant doubly even $[56,28,12]$ codes $D11, C_{56,1}, \dots, C_{56,5}$.
Generator matrices of $D11, C_{56,1}, \ldots, C_{56,5}$ are of the form~\eqref{eq:doubleCirculant} with first rows
\begin{align*}
  &(000101011011111000111111111), (000000000000110010101111011), \\
  &(000000001011011111110010111), (000000010011100111101110111), \\
  &(000000011011001001111101111), (000000101001111101011101011),
\end{align*}
respectively.

An $[n,k]$ code $C$ over $\FF_q$ is said to be an LCD code if $C \cap C^{\perp} = \{ \0_n \}$.
The concept of LCD codes was invented by Massey~\cite{massey1992linear}.
LCD codes have been applied in data storage, communication systems and cryptography.
For example, it is known that binary LCD codes can be used against side-channel attacks and fault injection attacks~\cite{carlet2016complementary}.
A binary LCD $[n,k]$ code is said to be optimal if it has the largest minimum weight among all binary LCD $[n,k]$ codes.
Massey~\cite{massey1992linear} gave the following characterization of LCD codes.
\begin{thm}[Massey~\cite{massey1992linear}]
    \label{thm:EucLCD}
    Let $C$ be an $[n,k]$ code over $\FF_q$ and let $G$ be a generator matrix of $C$. Then $C$ is an LCD code if and only if the $k \times k$ matrix $GG^T$ is nonsingular.
\end{thm}

The hull of a code $C$ is defined as $\Hull(C) = C \cap C^\perp$.
By definition, it follows that self-dual codes are codes with maximal hull dimension and LCD codes are codes with minimal hull dimension.
\begin{lem}[{\cite[Proposition 3.1]{guenda2018constructions}}] \label{lem:rankGGT}
    Let $C$ be an $[n,k]$ code over $\FF_q$ with generator matrix $G$.
    Then
    \begin{equation*}
        \rank(GG^T) = k - \dim(\Hull(C)).
    \end{equation*}
\end{lem}

\section{Construction method} \label{sec:constructMethod}
Let $C$ be an $[n,k]$ code over $\FF_q$ with generator matrix $\begin{pmatrix} I_k & A \\ \end{pmatrix}$ and let $x,y \in \FF_q^{n-k}$.
We denote by $r_i$ the $i$-th row of $A$.
Define an $n \times (n-k)$ matrix $A(x,y)$, where the $i$-th row $r'_i$ is defined as follows:
\begin{equation*}
    r'_i = r_i + (r_i,y)x - (r_i,x)y.
\end{equation*}
We denote by $C(A(x,y))$ the code with generator matrix $\begin{pmatrix} I_k & A(x,y) \\ \end{pmatrix}$.
\begin{rem}
    With the above notation, suppose that $x=\0_{n-k}$ or $y=\0_{n-k}$.
  Then it holds that $A(x,y)=A$.
  Hereafter, we assume that $x \neq \0_{n-k}$ and $y \neq \0_{n-k}$.
\end{rem}

\begin{thm} \label{F2:thm:LCP}
    Let $C$ be an $[n,k]$ code over $\FF_q$ with generator matrix $G = \begin{pmatrix} I_k & A \\ \end{pmatrix}$ and let $x,y \in \FF_q^{n-k}$.
    Suppose that $(x,x)=(y,y)=(x,y)=0$.
    Then $\dim(\Hull(C(A(x,y)))) = \dim(\Hull(C))$.
\end{thm}
\begin{proof}
    We denote by $r_i, r'_i$ the $i$-th rows of $A,A(x,y)$ respectively.
    It holds that
    \begin{align*}
        (r_i',r_j')
        &=(r_i+(r_i,y)x-(r_i,x)y,\ r_j+(r_j,y)x-(r_j,x)y) \\
        &=(r_i,r_j) +(r_j,y)(r_i,x) -(r_j,x)(r_i,y) +(r_i,y)(x,r_j) -(r_i,x)(y,r_j) \\
        &=(r_i,r_j).
    \end{align*}
    Therefore it follows that
    \begin{equation*}
        \begin{split}
            \begin{pmatrix} I_k & A(x,y) \\ \end{pmatrix} \begin{pmatrix} I_k & A(x,y) \\ \end{pmatrix}^T
            &= I_k + A(x,y)A(x,y)^T\\
            &= I_k + AA^T\\
            &= \begin{pmatrix} I_k & A \\ \end{pmatrix} \begin{pmatrix} I_k & A \\ \end{pmatrix}^T.
        \end{split}
    \end{equation*}
    By Lemma~\ref{lem:rankGGT}, the result follows.
\end{proof}

\begin{cor} \label{F2:thm:SO}
    Let $C$ be an $[n,k]$ code over $\FF_q$ with generator matrix $G = \begin{pmatrix} I_k & A \\ \end{pmatrix}$ and let $x,y \in \FF_q^{n-k}$.
    Suppose that $(x,x)=(y,y)=(x,y)=0$.
    Then $C(A(x,y))$ is a self-orthogonal code if and only if $C$ is a self-orthogonal code.
\end{cor}
\begin{proof}
    It holds that $C$ is a self-orthogonal code if and only if $\dim(\Hull(C))=k$.
    The result follows from Theorem~\ref{F2:thm:LCP}.
\end{proof}
\begin{rem} \label{rem:SO}
    With the notation of Corollary~\textup{\ref{F2:thm:SO}}, suppose that $C$ is a binary self-dual code with length $n \equiv 0 \pmod 4$.
    Then it follows that $(x,x)  = (x,\1_{n-k}) = (\1_{n-k},\1_{n-k}) = 0$ if $x$ has an even weight.
    Let $r_i, r'_i$ be the $i$-th rows of $A, A(x, \1_{n-k})$ respectively.
    Since $C$ is an even code, $\wt(r_i) \equiv 1 \pmod 2$ for all $1 \le i \le n$.
    Therefore we obtain the following:
    \begin{equation*}
        \begin{split}
            r_i' &= r_i + (r_i,\1_{n-k})x + (r_i,x)\1_{n-k}\\
                 &= r_i + x + (r_i,x)\1_{n-k},
        \end{split}
    \end{equation*}
    which shows that $A(x, \1_{n-k})$ is identical to $B'_{\Gamma'}$, the first case of~\textup{\cite[Theorem 2.2]{harada1996existence}}.
    Therefore Corollary~\textup{\ref{F2:thm:SO}} is a generalized version of the first case of~\textup{\cite[Theorem 2.2]{harada1996existence}}.
\end{rem}

\begin{cor} \label{F2:thm:LCD}
    Let $C$ be an $[n,k]$ code over $\FF_q$ with generator matrix $G = \begin{pmatrix} I_k & A \\ \end{pmatrix}$ and let $x,y \in \FF_q^{n-k}$.
    Suppose that $(x,x)=(y,y)=(x,y)=0$.
    Then $C(A(x,y))$ is an LCD code if and only if $C$ is an LCD code.
\end{cor}
\begin{proof}
    It holds that $C$ is an LCD code if and only if $\dim(\Hull(C))=0$.
    The result follows from Theorem~\ref{F2:thm:LCP}.
\end{proof}
\begin{rem} \label{rem:LCD}
    With the notation of Corollary~\textup{\ref{F2:thm:LCD}}, suppose that $C$ is a binary even LCD code and $n-k$ is even.
    Then it follows that $(x,x)  = (x,\1_{n-k}) = (\1_{n-k},\1_{n-k}) = 0$ if $x$ has an even weight.
    Let $r_i, r'_i$ be the $i$-th rows of $A, A(x, \1_{n-k})$ respectively.
    Since $C$ is an even code, $\wt(r_i) \equiv 1 \pmod 2$ for all $1 \le i \le n$.
    Therefore we obtain the following:
    \begin{equation*}
        \begin{split}
            r_i' &= r_i + (r_i,\1_{n-k})x + (r_i,x)\1_{n-k}\\
                 &= r_i + x + (r_i,x)\1_{n-k},
        \end{split}
    \end{equation*}
    which shows that $A(x, \1_{n-k})$ is identical to $A(x)$ in~\textup{\cite[Theorem 3.3]{harada2021construction}}.
    Therefore Corollary~\textup{\ref{F2:thm:LCD}} is a generalized version of~\textup{\cite[Theorem 3.3]{harada2021construction}}.
\end{rem}

\begin{lem}[{\cite[Theorem 1.4.3]{huffman2010fundamentals}}] \label{F2:lem:wt2}
    Let $x = (x_1, x_2, \ldots, x_n), y = (y_1, y_2, \ldots, y_n) \in \FF_2^n$.
    Then $\wt(x+y) = \wt(x) + \wt(y) - 2 \wt(x*y)$, where $x*y = (x_1y_1, x_2y_2, \ldots, x_ny_n)$.
\end{lem}
\begin{lem} \label{F2:lem:wt4}
    Let $C$ be a binary $[n,k]$ code with generator matrix $\begin{pmatrix} I_k & A \\ \end{pmatrix}$ and let $x,y \in \FF_q^{n-k}$.
    Suppose that $\wt(x) \equiv \wt(y) \equiv 0 \pmod{4}$ and $(x,y)=0$.
    Then $\wt(r_i') \equiv \wt(r_i) \pmod{4}$, where $r_i, r'_i$ denote the $i$-th rows of $A,A(x,y)$ respectively.
\end{lem}
\begin{proof}
    It holds that
    \begin{align*}
        \wt(r_i*((r_i,y)x-(r_i,x)y)) &=(r_i,\ (r_i,y)x-(r_i,x)y)\\
        &=(r_i,y)(r_i,x)-(r_i,x)(r_i,y)\\
        &=0,\\
        \wt((r_i,y)x*(r_i,x)y) &=((r_i,y)x,(r_i,x)y)\\
        &=(r_i,y)(r_i,x)(x,y)\\
        &\equiv 0 \pmod 2,
    \end{align*}
    where we regard $r_i, x, y$ as vectors in $\ZZ^{n-k}$.
    Therefore, by Lemma~\ref{F2:lem:wt2}, it follows that
    \begin{align*}
        \wt(r_i')
        &=\wt(r_i+(r_i,y)x-(r_i,x)y) \\
        &=\wt(r_i)+\wt((r_i,y)x-(r_i,x)y) - 2\wt(r_i*((r_i,y)x-(r_i,x)y))\\
        &=\wt(r_i)+\wt((r_i,y)x)+\wt((r_i,x)y) - 2\wt((r_i,y)x*(r_i,x)y)\\
        & \equiv \wt(r_i) \pmod{4}.
    \end{align*}
    This completes the proof.
\end{proof}

\begin{thm} \label{F2:thm:DE}
    Let $C$ be a binary $[n,k]$ code with generator matrix $\begin{pmatrix} I_k & A \\ \end{pmatrix}$ and let $x,y \in \FF_2^{n-k}$.
    Suppose that $\wt(x) \equiv \wt(y) \equiv 0 \pmod 4$ and $(x,y) = 0$.
    Then $C(A(x,y))$ is a doubly even code if and only if $C$ is a doubly even code.
\end{thm}
\begin{proof}
    We denote by $r_i, r'_i$ the $i$-th rows of $A,A(x,y)$ respectively.
    Suppose that $C$ is a doubly even code.
    Then, by Corollary~\ref{F2:thm:SO} and the second part of Lemma~\ref{F2:lem:doublyevencode}, $C(A(x,y))$ is a self-orthogonal code.
    Furthermore, by Lemma~\ref{F2:lem:wt4}, $\wt(r'_i) \equiv \wt(r_i) \equiv 3 \pmod 4$ follows for all $1 \le i \le n$.
    Therefore it holds that $C(A(x,y))$ is a doubly even code by the first part of Lemma~\ref{F2:lem:doublyevencode}  
    By the same argument, the converse holds.
\end{proof}
\begin{rem} \label{rem:DE}
    With the notation of Theorem~\textup{\ref{F2:thm:DE}}, suppose that $C$ is a binary doubly even self-dual code.
    Then it follows that $\wt(x) \equiv \wt(\1_{n-k}) \equiv 0 \pmod 4$ and $(x,\1_{n-k})=0$ if $x$ has a weight divisible by four.
    Let $r_i, r'_i$ be the $i$-th rows of $A, A(x, \1_{n-k})$ respectively.
    Since $C$ is doubly even, $\wt(r_i) \equiv 3 \pmod 4$ for all $1 \le i \le n$.
    Therefore we obtain the following:
    \begin{equation*}
        \begin{split}
            r_i' &= r_i + (r_i,\1_{n-k})x + (r_i,x)\1_{n-k}\\
                 &= r_i + x + (r_i,x)\1_{n-k},
        \end{split}
    \end{equation*}
    which shows that $A(x, \1_{n-k})$ is identical to $A'_{\Gamma}$ in~\textup{\cite[Theorem 2.2]{harada1995new}}, where $m$ is even.
    Therefore Theorem~\textup{\ref{F2:thm:DE}} is a generalized version of~\textup{\cite[Theorem 2.2]{harada1995new}}, where $m$ is even.
\end{rem}

\section{Basic properties} \label{sec:basicProperties}
In this section, we state basic properties of the construction method given in Section~\ref{sec:constructMethod}.
Let $x,y \in \FF_q^m$.
Define the following $m \times m$ matrix $M(x,y)$:
\begin{equation*}
    M(x,y) = I_{m} + y^Tx - x^Ty.
\end{equation*}

\begin{lem} \label{F2:prp:AMxy}
    Let $C$ be an $[n,k]$ code over $\FF_q$ with generator matrix $\begin{pmatrix} I_k & A \\ \end{pmatrix}$ and let $x,y \in \FF_q^{n-k}$.
    Then it holds that
    \begin{equation*}
        A(x, y) = A M(x,y).
    \end{equation*}
\end{lem}
\begin{proof}
    For a matrix $M$, we denote by $M_{i,j}$ the $(i,j)$-entry of $M$.
    By definition it follows that
    \[
        (A M(x,y))_{i,j} = \sum_{l=1}^{n-k} A_{i,l} M(x,y)_{l, j}.
    \]
    On the other hand, it holds that
    \begin{align*}
        r'_i &= r_i + (r_i,y)x - (r_i,x)y\\
             &= r_i (I_{n-k} + y^Tx - x^Ty)\\
             &= (A_{i,1}, A_{i,2}, \ldots, A_{i,n}) M(x,y)\\
             &= (\sum_{l=1}^{n-k} A_{i,l} M(x,y)_{l,1},\sum_{l=1}^{n-k} A_{i,l} M(x,y)_{l,2}, \dots,\sum_{l=1}^{n-k} A_{i,l} M(x,y)_{l,n}), 
    \end{align*}
    where $r_i,r'_i$ denote the $i$-th rows of $A,A(x,y)$ respectively.
    Therefore it holds that
    \[
        A(x,y)_{i,j} = \sum_{l=1}^{n-k} A_{i,l} M(x,y)_{l, j}.
    \]
    This completes the proof.
\end{proof}

\begin{thm} \label{F2:lem:Mxy} 
    Let $C$ be an $[n,k]$ code over $\FF_q$ with generator matrix $\begin{pmatrix} I_k & A \\ \end{pmatrix}$ and let $x,y \in \FF_q^{n-k}$.
    Suppose that $(x, x) = (y, y) = (x, y) = 0$.
    Then the following holds:
    \begin{align}
        &C(A(x,y)) = C(A(-x,-y)), \label{F2:eq:Mxy}\\
        &C(A(y,x)) = C(A(x,-y)) = C(A(-x,y)). \label{F2:eq:Myx}
    \end{align}
\end{thm}
\begin{proof}
    For~\eqref{F2:eq:Mxy}, it holds that
    \begin{align*}
        M(-x,-y) &= I_m + (-y)^T(-x) - (-x)^T(-y)\\
        &= I_m + y^Tx - x^Ty\\
        &= M(x,y).
    \end{align*}
    For~\eqref{F2:eq:Myx}, the following holds:
    \begin{align*}
        M(x,-y) &= I_m - y^Tx + x^Ty\\
        &= I_m + x^Ty - y^Tx\\
        &= M(y,x),\\
        M(-x,y) &= I_m - y^Tx + x^Ty\\
        &= I_m + x^Ty - y^Tx\\
        &= M(y,x).
    \end{align*}
    This completes the proof.
\end{proof}
In the following sections, we apply the construction method only to binary codes.
However, for codes over $\FF_q\ (q \ge 3)$, Theorem~\ref{F2:lem:Mxy} reduces computations.

\section{Extremal binary doubly even $[56,28,12]$ codes} \label{sec:SD}
In this section we are concerned only with binary codes.
Therefore we omit the term ``binary''.
Bhargava, Young and Bhargava~\cite{bhargava1981characterization} constructed an extremal doubly even $[56,28,12]$ code.
Yorgov~\cite{yorgov1987method} proved that there exist exactly $16$ inequivalent extremal doubly even $[56,28,12]$ codes with automorphisms of order $13$.
As stated in~\cite{harada2006self}, one of the $16$ codes in~\cite{yorgov1987method} is equivalent to the code in~\cite{bhargava1981characterization}.
Bussemarker and Tonchev~\cite{bussemarker1989new} constructed $6$ extremal doubly even $[56,28,12]$ codes.
As stated in~\cite{harada2006self}, the first code among the six codes in~\cite{bussemarker1989new} is equivalent to the code in~\cite{bhargava1981characterization}.
Moreover, Kimura~\cite{kimura1994extremal} showed that $5$th and $6$th codes in~\cite{bussemarker1989new} are equivalent.
Harada~\cite{harada1996existence} constructed $137$ inequivalent extremal doubly even $[56,28,12]$ codes.
Harada, Gulliver and Kaneta~\cite{harada1998classification} showed that there exist exactly nine inequivalent extremal double circulant doubly even $[56,28,12]$ codes.
Harada~\cite{harada2006self} constructed $1122$ inequivalent extremal doubly even $[56,28,12]$ codes.
This result is a generalization of~\cite{harada1996existence} and for any code $C$ in~\cite{harada1996existence} there exists a code $C'$ in~\cite{harada2006self} such that $C \simeq C'$.
Yankov and Russeva~\cite{yankov2011binary} proved that there exist exactly $4202$ inequivalent extremal doubly even $[56,28,12]$ codes having automorphisms of order $7$.
As stated in~\cite{yankov2011binary}, one of the $4202$ codes in~\cite{yankov2011binary} is equivalent to a code in~\cite{harada1996existence}.
Yankov and Lee~\cite{yankov2014newbinary} proved that there exist exactly $3763$ inequivalent extremal $[56,28,12]$ codes having automorphisms of order $5$.
Therefore the number of previously known extremal doubly even $[56,28,12]$ codes is $9115$, as stated in~\cite[Proposition 7]{yankov2014newbinary}.

In~\cite{harada2006self}, Harada applied~\cite[Theorem 2.2]{harada1995new} to six inequivalent extremal double circulant doubly even $[56,28,12]$ codes $D11, C_{56,1}, \ldots, C_{56,5}$.
As stated earlier, Theorem~\ref{F2:thm:DE} is a generalized version of~\cite[Theorem 2.2]{harada1995new}.
In this section, we apply Theorem~\ref{F2:thm:DE} to $D11, C_{56,1}, \dots, C_{56,5}$ and construct $661$ new inequivalent extremal doubly even $[56,28,12]$ codes.
This illustrates the effectiveness of Theorem~\ref{F2:thm:DE}.

In order to illustrate our method, we consider the code $D11$ as an example.
Let $y_i$ denote the vector of length $28$ such that $y_i=(0,\dots,0,1,\dots,1), \wt(y_i) = i$.
Applying Theorem~\ref{F2:thm:DE} to $y_4$ and all $x \in \FF_2^{28}$ such that $x \ne \0_{28}$ and $(x,x)=(x,y_4)=0$, we constructed $45$ inequivalent extremal doubly even $[56,28,12]$ codes. 
Furthermore, applying Theorem~\ref{F2:thm:DE} to $y_i\ (i=8,12,16,20,24)$ and all $x \in \FF_2^{28}$ such that $x \ne \0_{28}$ and $(x,x)=(x,y_i)=0$, we constructed $45,19,15,2,33,4$ inequivalent extremal doubly even $[56,28,12]$ codes respectively.

By the following method, we verified that the above codes are all inequivalent.
For an extremal doubly even $[56,28,12]$ code $C$, we define $M=(m_{i,j})$ to be an $8196 \times 56$ matrix whose rows composed of codewords of $C$ with weight $12$.
Furthermore, for a positive integer $t$, we define 
\begin{equation*}
    N_t = \# \left\{ \{ j_1,j_2,j_3,j_4 \} \in  {\displaystyle { 56 \choose 4 }}  \mid \sum_{i=1}^{8196} m_{i,j_1} m_{i,j_2} m_{i,j_3} m_{i,j_4} = t \right\}.
\end{equation*}
Harada~\cite{harada2006self} showed that two extremal doubly even $[56,28,12]$ codes $C_1,C_2$ are inequivalent if the sequences $(N_1, N_2, \dots)$ constructed from $C_1,C_2$ are distinct.
According to this result, we compared the sequence $(N_1, N_2, \dots, N_{56})$ for the classification.
Consequently we found no pair of codes whose sequences are identical.
Therefore we verified that the number of inequivalent codes constructed from $D11$ is $118$.
By the same method, we constructed inequivalent extremal doubly even $[56,28,12]$ codes from $C_{56,i}\ (i=1,2,\dots,5)$.
In Table~\ref{table:inequiv} we show the number of inequivalent codes constructed by this method.
We denote the inequivalent codes constructed from $D11,C_{56,1},\dots,C_{56,5}$ by $D_i\ (i=1,2,\dots,118), E_i\ (i=1,2,\dots,56), F_i\ (i=1,2,\dots,105), G_i\ (i=1,2,\dots,59), H_i\ (i=1,2,\dots,212), K_i\ (i=1,2,\dots,115)$ respectively.
The $x,y$ in Corollary~\ref{F2:thm:DE} for all codes we constructed can be obtained electronically from \url{https://www.math.is.tohoku.ac.jp/~mharada/Ishizuka/56.txt}.
\begin{table}[thbp]\caption{Inequivalent extremal doubly even $[56,28,12]$ codes}\label{table:inequiv}
\begin{center}
\begin{tabular}{c|ccccccccccccccccccccccccccccccccccc}
       & $y_4$ & $y_8$ & $y_{12}$ & $y_{16}$ & $y_{20}$ & $y_{24}$ & total \\
\hline
     $D11$ & $45$ & $19$ & $15$ & $2$ & $33$ & $4$ & $118$\\
     $C_{56,1}$ & $16$ & $3$ & $1$ & $0$ & $10$ & $26$ & $56$\\
     $C_{56,2}$ & $34$ & $27$ & $26$ & $1$ & $3$ & $14$ & $105$\\
     $C_{56,3}$ & $10$ & $0$ & $23$ & $2$ & $0$ & $24$ & $59$\\
     $C_{56,4}$ & $10$ & $109$ & $17$ & $58$ & $2$ & $16$ & $212$\\
     $C_{56,5}$ & $17$ & $53$ & $25$ & $11$ & $5$ & $4$ & $115$\\
\end{tabular}
\end{center}
\end{table}
Comparing sequences $(N_1, N_2, \dots, N_{56})$, we found that there exist four pairs of codes $(D_{115}, K_{112}), (D_{116}, K_{113}), (D_{117}, K_{114}), (D_{118}, K_{115})$ whose sequences are identical.
By the {\sc Magma} function \texttt{IsIsomorphic}, we verified that two codes of all the four pairs are equivalent.
Therefore the number of inequivlent extremal doubly even $[56,28,12]$ codes constructed as above is $661$. 

Finally, we verified that the $661$ codes are inequivalent to any of the previously known extremal doubly even $[56,28,12]$ codes as follows:
By the {\sc Magma} function \texttt{AutomorphismGroup}, we verified that the $661$ codes have automorphism groups of order $1$.
Consequently it follows that the $661$ codes are inequivalent to any of the codes in~\cite{yankov2014newbinary},~\cite{yankov2011binary},~\cite{yorgov1987method}.
Furthermore we verified that all the codes except $D_{118}$ have sequences $(N_1, N_2, \ldots, N_{56})$ different from that of any code in~\cite{bhargava1981characterization},~\cite{bussemarker1989new},~\cite{harada2006self},~\cite{harada1998classification},~\cite{yorgov1987method}.
The sequence $(N_1,N_2,\dots,N_{56})$ of $D_{118}$ is identical to that of the $25$th code constructed from $C_{56,2}$ in~\cite{harada2006self}.
However we verified by the Magma function \texttt{IsIsomorphic} that the two codes are inequivalent.
Consequently it follows that the $661$ codes are inequivalent to any of the codes in~\cite{bhargava1981characterization},~\cite{bussemarker1989new},~\cite{harada2006self},~\cite{harada1998classification},~\cite{yorgov1987method}.
As stated in the beginning of this section, the number of the previously known inequivalent doubly even $[56,28,12]$ codes is $9115$.
Therefore we have Proposition~\ref{prp:newCode}.
\begin{prp} \label{prp:newCode}
    There exist at least $9776$ inequivalent extremal doubly even $[56,28,12]$ codes.
\end{prp}

\section{Optimal binary LCD codes of length $n=26, 28 \le n \le 40$}
\label{sec:LCD}
In this section we are concerned only with binary codes.
Therefore we omit the term ``binary''.
Let $d_{LCD}(n,k)$ denote the largest minimum weight among all LCD $[n,k]$ code.
Galvez, Kim, Lee, Roe and Won~\cite{galvez2018somebounds}, Harada and Saito~\cite{harada2019binary}, Araya and Harada~\cite{araya2020minimum} determined the exact value of $d_{LCD}(n,k)$ for $n \le 12, 13 \le n \le 16, 17 \le n \le 24$ respectively.
Bouyuklieva~\cite{bouyuklieva2020optimal} determined the exact value of $d_{LCD}(n,k)$ for $n=25,27$ and gave $d_{LCD}(n,k)$ for $n=26, 28 \le n \le 40$.
Galvez, Kim, Lee, Roe and Won~\cite{galvez2018somebounds}, Harada and Saito~\cite{harada2019binary}, Araya and Harada~\cite{araya2020minimum}, Araya, Harada and Saito~\cite{araya2021characterization} determined the exact value of $d_{LCD}(n,k)$ for $k=2,3,4,5$ respectively.
Also, Dougherty, Kim, Ozkaya, Sok and Sol\'{e}~\cite{dougherty2017combinatorics}, Araya and Harada~\cite{araya2021classification}, Araya, Harada and Saito~\cite{araya2021minimum} determined the exact value of $d_{LCD}(n,k)$ for $k=n-1$, $k \in \{n-2,n-3,n-4\}$, $k=n-5$ respectively.
For all $n=26, 28 \le n \le 40$, Bouyuklieva~\cite{bouyuklieva2020optimal} determined the exact value of $d_{LCD}(n,k)$ for $5 \le k \le 8$.

Recently Harada~\cite{harada2021construction} constructed $15$ optimal LCD codes by~\cite[Theorem 3.3]{harada2021construction}.
As stated earlier, Corollary~\ref{F2:thm:LCD} is a generalized version of~\cite[Theorem 3.3]{harada2021construction}.
In this section, we apply Corollary~\ref{F2:thm:LCD} in order to improve some of the previously known lower bounds on $d_{LCD}(n,k)$ for $n=26, 28 \le n \le 40$ and $9 \le k \le n - 6$.
For a code $C$, we denote by $C^T,C_T$ the punctured, the shortened codes of $C$ on a set of coordinates $T$ respectively.
In this section, shortened codes, punctured codes were constructed by the {\sc Magma} functions \texttt{ShortenCode}, \texttt{PunctureCode} respectively.

In order to obtain lower bounds, we use the following method:
First, by the {\sc Magma} function \texttt{BestKnownLinearCode}, we obtained a $[49,32,7]$ code $C_{49,32,7}$, a $[42,14,13]$ code $C_{42,14,13}$ and a $[51,28,9]$ code $C_{51,28,9}$.
Generator matrices of these codes can be obtained electronically from \url{https://www.math.is.tohoku.ac.jp/~mharada/Ishizuka/generator.txt}.
Then we verified that $((C_{49,32,7})_{S_1})^{P_1}$, $((C_{42,14,13})_{S_2})^{P_2}$, $((C_{51,28,9})_{S_3})^{P_3}$ are an LCD $[37,22,5]$ code, an LCD $[38,13,11]$ code, an LCD $[40,22,6]$ code respectively, where the set of coordinates $P_i,\ S_i\ (i=1,2,3)$ are given in Table~\ref{table:PS}.
Define $A_{37,22,5}$, $A_{38,13,10}$, $A_{40,22,6}$ as in Figure~\ref{fig:genmat}.
Then $\begin{pmatrix} I_{22} & A_{37,22,5} \\ \end{pmatrix}$, $\begin{pmatrix} I_{13} & A_{38,13,10} \\ \end{pmatrix}$, $\begin{pmatrix} I_{22} & A_{40,22,6} \\ \end{pmatrix}$ are generator matrices of the LCD $[37,22,5]$, the LCD $[38,13,10]$, the LCD $[40,22,6]$ codes respectively.
Applying Corollary~\ref{F2:thm:LCD} to $\begin{pmatrix} I_{22} & A_{37,22,5} \\ \end{pmatrix}$, $\begin{pmatrix} I_{13} & A_{38,13,10} \\ \end{pmatrix}$, $\begin{pmatrix} I_{22} & A_{40,22,6} \\ \end{pmatrix}$, we found an LCD $[37,22,6]$ code $C_{37,22,6}$, an LCD $[38,13,11]$ code $C_{38,13,11}$, an LCD $[40,22,7]$ code $C_{40,22,7}$ respectively.
The vectors $x,y$ in Corollary~\ref{F2:thm:LCD} are listed in Table~\ref{table:xyCnkd}. 
Therefore we obtain Proposition~\ref{prp:LCD:PI}.
\begin{prp} \label{prp:LCD:PI}
    \begin{enumerate}
        \item There exists an LCD $[37,22,6]$ code.
        \item There exists an LCD $[38,13,11]$ code.
        \item There exists an LCD $[40,22,7]$ code.
    \end{enumerate}
\end{prp}
\begin{table}[thbp]\caption{$P_i,S_i$ for $i=1,2,3$}\label{table:PS}
\begin{center}
\begin{tabular}{c|ccccccccccccccccccccccccccccccccccc}
    $i$ & $P_i$ & $S_i$\\
\hline
    $1$ & $\{2,5\}$ & $\{1,2,3,4,5,6,7,8,9,11\}$ \\
    $2$ & $\{1,2,4\}$ & $\{3\}$ \\
    $3$ & $\{1,4,5,6,7\}$ & $\{1,2,3,4,5,7\}$ \\
\end{tabular}
\end{center}
\end{table}

\begin{table}[thbp]\caption{$C_{n,k,d}$ with $x,y$}\label{table:xyCnkd}
\begin{center}
\begin{tabular}{c|ccccccccccccccccccccccccccccccccccc}
    $C_{n,k,d}$ & $x$ & $y$\\
\hline
    $C_{37,22,6}$ & $(0 1 0 1 1 0 0 1 1 0 1 1 1 1 1)$ & $(1 1 0 0 1 0 1 1 0 0 0 0 0 0 1)$ \\
    $C_{38,13,11}$ & $(0 0 1 0 0 1 1 1 0 0 1 1 0 0 0 1 0 1 1 0 0 0 1 0 0)$ & $ (1 1 1 0 1 0 0 0 0 1 1 1 0 1 0 1 1 1 0 0 0 1 1 0 1)$ \\
    $C_{40,22,7}$ & $(1 0 1 1 1 1 0 1 1 0 1 1 0 1 0 0 1 1)$ & $(0 0 1 0 1 1 1 0 0 0 1 1 1 0 0 0 0 1)$
\end{tabular}
\end{center}
\end{table}

From the previously known results of $d_{LCD}(n,k)$ described in the beginning of this section, we are concerned only with $d_{LCD}(n,k)$ for $n=26, 28 \le n \le 40$ and $9 \le k \le n-6$.
Let $d_K(n,k)$ denote the largest minimum weight among currently known $[n,k]$ codes.
By the {\sc Magma} function {\tt BestKnownLinearCode}, one can construct an $[n,k,d_K(n,k)]$ code
for all $n=26, 28 \le n \le 40$ and $9 \le k \le n-6$.
In addition, by considering shortened codes and punctured codes of $[n,k,d_K(n,k)]$ codes, we found LCD $[n,k,d_K(n,k)]$ codes for
\begin{equation}\label{eq:LCD2}
\begin{split}
    (n,k,d_K(n,k))=&(29,11,9),(30,12,9),(31,11,10),(31,12,10),(31,13,9),\\
    &(31,21,5),(32,12,10),(32,22,5),(33,23,5),(34,9,13),\\
    &(34,13,10),(34,14,10),(35,22,6),(35,24,5),(36,15,10),\\
    &(36,16,10),(36,22,6),(36,24,6),(36,25,5),(37,23,6),\\
    &(37,24,6),(37,26,5),(38,24,6),(38,25,6),(38,26,6),\\
    &(38,27,5),(39,18,10),(39,25,6),(39,26,6),(39,28,5),\\
    &(40,26,6),(40,28,6),(40,29,5).
    \end{split}
\end{equation}    
Consequently we obtain Proposition~\ref{prp:LCD:trivial}.
\begin{prp} \label{prp:LCD:trivial}
    There exists an optimal LCD $[n,k,d]$ code for $(n,k,d)$ listed in \eqref{eq:LCD2}.
\end{prp}
By a method similar to that given in the above, we found LCD $[n,k,d_K(n,k)-1]$ codes and LCD $[n,k,d_K(n,k)-2]$ codes for
\begin{equation}\label{eq:LCD3}
\begin{split}
    (n,k,d_K(n,k)-1)=&(30,11,9),(31,15,7),(32,13,9),(32,15,7),(32,16,7),\\
    &(32,21,5),(33,14,9),(33,15,8),(33,16,7),(33,21,5),\\
    &(34,15,9),(34,16,8),(34,17,7),(34,23,5),(35,9,13),\\
    &(35,16,9),(35,17,7),(35,18,7),(35,21,5),(35,23,5),\\
    &(36,17,8),(36,18,7),(36,19,7),(36,21,6),(37,17,9),\\
    &(37,18,8),(37,19,7),(37,20,7),(37,25,5),(38,10,13),\\
    &(38,17,9),(38,18,9),(38,19,8),(38,20,7),(38,21,7),\\
    &(38,23,6),(39,11,13),(39,14,11),(39,17,10),(39,19,9),\\
    &(39,20,8),(39,21,7),(39,22,7),(39,24,6),(39,27,5),\\
    &(39,33,2),(40,9,15),(40,12,13),(40,18,10),(40,19,9),\\
    &(40,20,9),(40,23,7),(40,25,6),(40,27,5),(40,33,3),\\
    &(40,34,2).
\end{split}
\end{equation}    
\begin{equation}\label{eq:LCD4}
\begin{split}
    (n,k,d_K(n,k)-2)=&(37,21,6),(38,22,6),(39,13,11),(39,23,6),(40,11,13),\\
    &(40,13,12),(40,14,11),(40,17,10),(40,21,7),(40,24,6).
\end{split}
\end{equation}    

Consequently we obtain Proposition~\ref{prp:LCD:trivial1}.
\begin{prp} \label{prp:LCD:trivial1}
    There exists an LCD $[n,k,d]$ code for $(n,k,d)$ listed in \eqref{eq:LCD3} and~\eqref{eq:LCD4}.
\end{prp}
In Tables~\ref{Tab:F2-0} through~\ref{Tab:F2-2}, we give $d_{LCD}(n,k)$ for $n=26,28 \le n \le 40$ and $9 \le k \le n-6$.
In order to obtain upper bounds, we use the following:
\begin{equation*}
    d_{LCD}(n,k) \le d(n,k),
\end{equation*}
where $d(n,k)$ denotes the largest minimum weight among all $[n,k]$ codes.
The values of $d(n,k)$ are given in~\cite{grassl-codetables}.
For the parameters listed in Proposition~\ref{prp:LCD:PI}, we mark $d_{LCD}(n,k)$ by $*$ in Tables~\ref{Tab:F2-0} through~\ref{Tab:F2-2}.
Furthermore, for the parameters given in Propositions~\ref{prp:LCD:PI} through~\ref{prp:LCD:trivial1}, we give $d_{LCD}(n,k)$ in boldface.
For each of the parameters, an LCD code can be obtained electronically from \url{https://www.math.is.tohoku.ac.jp/~mharada/Ishizuka/LCD.txt}.
\begin{table}[thbp]\caption{$d_{LCD}(n,k)$, where $26 \le n \le 40, 9 \le k \le 17$}\label{Tab:F2-0}\begin{center}{\small\begin{tabular}{c|cccccccccccccccccccccccccccccccccccccccccccccccc}\noalign{\hrule height 0.8pt}
$n\backslash k$ & 9 & 10 & 11 & 12 & 13 & 14 & 15 & 16 & 17\\
26 & 9 & 8 & 8 & 8 & 7 & 6 & 5--6 & 5 & 4\\
27 & 9 & 9 & 8 & 8 & 7 & 6 & 6 & 6 & 5\\
28 & 10 & 10 & 8 & 8 & 8 & 7 & 6 & 6 & 5--6\\
29 & 10 & 10 & \textbf{9} & 8 & 8 & 8 & 6 & 6 & 6\\
30 & 11 & 10 & \textbf{9}--10 & \textbf{9} & 8 & 8 & 6--7 & 6 & 6\\
31 & 11 & 10 & \textbf{10} & \textbf{10} & \textbf{9} & 8 & \textbf{7}--8 & 6--7 & 6\\
32 & 12 & 11 & 10 & \textbf{10} & \textbf{9}--10 & 8--9 & \textbf{7}--8 & \textbf{7}--8 & 6--7\\
33 & 12 & 12 & 10--11 & 10 & 9--10 & \textbf{9}--10 & \textbf{8}--9 & \textbf{7}--8 & 6--8\\
34 & \textbf{13} & 12 & 11--12 & 10--12 & \textbf{10} & \textbf{10} & \textbf{9}--10 & \textbf{8}--9 & \textbf{7}--8\\
35 & \textbf{13}--14 & 12--13 & 12 & 10--12 & 10--11 & 10 & 9--10 & \textbf{9}--10 & \textbf{7}--8\\
36 & 13--14 & 12--14 & 12--13 & 11--12 & 10--12 & 10--11 & \textbf{10} & \textbf{10} & \textbf{8}--9\\
37 & 13--15 & 12--14 & 12--14 & 12--13 & 10--12 & 10--12 & 10--11 & 10 & \textbf{9}--10\\
38 & 14--16 & \textbf{13}--14 & 12--14 & 12--14 & \textbf{11*}--12 & 10--12 & 10--12 & 10--11 & \textbf{9}--10\\
39 & 14--16 & 14--15 & \textbf{13}--14 & 12--14 & \textbf{11}--13 & \textbf{11}--12 & 10--12 & 10--12 & \textbf{10}--11\\
40 & \textbf{15}--16 & 14--16 & \textbf{13}--15 & \textbf{13}--14 & \textbf{12}--14 & \textbf{11}--13 & 10--12 & 10--12 & \textbf{10}--12\\
\noalign{\hrule height 0.8pt}\end{tabular}}\end{center}\end{table}

\begin{table}[thbp]\caption{$d_{LCD}(n,k)$, where $26 \le n \le 40, 18 \le k \le 26$}\label{Tab:F2-1}\begin{center}{\small\begin{tabular}{c|cccccccccccccccccccccccccccccccccccccccccccccccc}\noalign{\hrule height 0.8pt}
$n\backslash k$ & 18 & 19 & 20 & 21 & 22 & 23 & 24 & 25 & 26\\
26 & 4 & 4 & 4 &  &  &  &  &  & \\
27 & 4 & 4 & 4 & 3 &  &  &  &  & \\
28 & 5 & 4 & 4 & 4 & 3 &  &  &  & \\
29 & 6 & 5 & 4 & 4 & 4 & 3 &  &  & \\
30 & 6 & 5 & 5 & 4 & 4 & 4 & 3 &  & \\
31 & 6 & 6 & 6 & \textbf{5} & 4 & 4 & 4 & 3 & \\
32 & 6 & 6 & 6 & \textbf{5}--6 & \textbf{5} & 4 & 4 & 3--4 & 3\\
33 & 6--7 & 6 & 6 & \textbf{5}--6 & 6 & \textbf{5} & 4 & 4 & 4\\
34 & 6--8 & 6--7 & 6 & 6 & 6 & \textbf{5}--6 & 4 & 4 & 4\\
35 & \textbf{7}--8 & 6--8 & 6--7 & \textbf{5}--6 & \textbf{6} & \textbf{5}--6 & \textbf{5} & 4 & 4\\
36 & \textbf{7}--8 & \textbf{7}--8 & 6--8 & \textbf{6}--7 & \textbf{6} & 6 & \textbf{6} & \textbf{5} & 4\\
37 & \textbf{8}--9 & \textbf{7}--8 & \textbf{7}--8 & \textbf{6}--8 & \textbf{6*}--7 & \textbf{6} & \textbf{6} & \textbf{5}--6 & \textbf{5}\\
38 & \textbf{9}--10 & \textbf{8}--9 & \textbf{7}--8 & \textbf{7}--8 & \textbf{6}--8 & \textbf{6}--7 & \textbf{6} & \textbf{6} & \textbf{6}\\
39 & \textbf{10} & \textbf{9}--10 & \textbf{8}--9 & \textbf{7}--8 & \textbf{7}--8 & \textbf{6}--8 & \textbf{6}--7 & \textbf{6} & \textbf{6}\\
40 & \textbf{10}--11 & \textbf{9}--10 & \textbf{9}--10 & \textbf{7}--9 & \textbf{7*}--8 & \textbf{7}--8 & \textbf{6}--8 & \textbf{6}--7 & \textbf{6}\\
\noalign{\hrule height 0.8pt}\end{tabular}}\end{center}\end{table}

\begin{table}[thbp]\caption{$d_{LCD}(n,k)$, where $33 \le n \le 40, 27 \le k \le 34$}\label{Tab:F2-2}\begin{center}{\small\begin{tabular}{c|cccccccccccccccccccccccccc}\noalign{\hrule height 0.8pt}
$n\backslash k$ & 27 & 28 & 29 & 30 & 31 & 32 & 33 & 34\\
33 & 3 &  &  &  &  &  &  & \\
34 & 3--4 & 3 &  &  &  &  &  & \\
35 & 4 & 4 & 3 &  &  &  &  & \\
36 & 4 & 4 & 3--4 & 3 &  &  &  & \\
37 & 4 & 4 & 4 & 4 & 3 &  &  & \\
38 & \textbf{5} & 4 & 4 & 4 & 3--4 & 3 &  & \\
39 & \textbf{5}--6 & \textbf{5} & 4 & 4 & 4 & 4 & \textbf{2}--3 & \\
40 & \textbf{5}--6 & \textbf{6} & \textbf{5} & 4 & 4 & 4 & \textbf{3}--4 & \textbf{2}--3\\
\noalign{\hrule height 0.8pt}\end{tabular}}\end{center}\end{table}

\section*{Acknowledgement}
The authors would like to thank supervisor Professor Masaaki Harada for introducing the problem, and his helpful advice and encouragement.

\begin{landscape}
\begin{figure}[htbp]
\begin{center}
\begin{equation*}
A_{37,22,5}=\begin{pmatrix}
000110111001011\\
100100110001100\\
011100010000100\\
001110001000010\\
001001001100011\\
111101001100101\\
000011111101010\\
100101001111001\\
101011001101000\\
111111011111010\\
011111101111101\\
111000010101000\\
111000111011000\\
010010010101110\\
000111000010101\\
111010001011110\\
110111111100001\\
101100011100110\\
110010111111111\\
100000110101011\\
001101000001101\\
010101110011100\\
\end{pmatrix}
A_{38,13,10}=\begin{pmatrix}
1101001110101100011010111\\
0010011101001101111011001\\
1000111010010001010001001\\
1101101001111111000100001\\
1111000000001000001110101\\
1110010100110011101011111\\
1110111110101110011001010\\
1010010001111011010110010\\
1010011011011100001010111\\
1110100100010100010010111\\
1110100110111101100101110\\
0111010011011110110010111\\
0111010011110100101111001\\
\end{pmatrix}
A_{40,22,6}=\begin{pmatrix}
010111001011101001\\
011100100111010110\\
001011001011101100\\
110010010010001001\\
000110111001101111\\
000001100100101101\\
001111111001011111\\
100101000000000011\\
101011011001101010\\
010100011111010000\\
101100110000010100\\
011010010101010110\\
011011011101000100\\
001010001100001010\\
010100000000111001\\
110110011100011011\\
111000011001011111\\
000000110101101010\\
101101110011011100\\
001001010100000011\\
010110110100111100\\
001001101110010111\\
\end{pmatrix}
\end{equation*}
\caption{Matrices $A_{37,22,5},A_{38,13,10},A_{40,22,6}$}
\label{fig:genmat}
\end{center}
\end{figure}
\end{landscape}


\begin{thebibliography}{99}
    \bibitem{araya2020minimum} M.\,Araya and M.\,Harada, On the minimum weights of binary linear complementary dual codes, Cryptogr.\ Commun.\ 12 (2020), 285--300.
    \bibitem{araya2021classification} M.\,Araya and M.\,Harada, On the classification of quaternary optimal Hermitian LCD codes, arXiv:2011.04139.
    \bibitem{araya2021characterization} M.\,Araya, M.\,Harada and K.\,Saito, Characterization and classification of optimal LCD codes, Des.\ Codes\ Cryptogr.\ 89 (2021), 617--640.
    \bibitem{araya2021minimum} M.\,Araya, M.\,Harada and K.\,Saito, On the minimum weights of binary LCD codes and ternary LCD codes, Finite\ Fields\ Appl.\ (to appear), arXiv:1908.08661.
    \bibitem{bhargava1981characterization} V.\,K.\,Bhargava, G.\,Young and A.\,K.\,Bhargava, A characterization of a $(56,28)$ extremal self-dual code, IEEE Trans.\ Inform.\ Theory\ 27 (1981), 258--260.
    \bibitem{bosma1997magma} W.\,Bosma, J.\,Cannon and C.\,Playoust, The Magma algebra system.\ I. The user language, J. Symbolic Comput.\ 24 (1997), 235--265.
    \bibitem{bouyuklieva2020optimal} S.\,Bouyuklieva, Optimal binary LCD codes, arXiv:2010.13399.
    \bibitem{bussemarker1989new} F.\,C.\,Bussemarker and V.\,D.\,Tonchev, New extremal doubly-even codes of length $56$ derived from Hadamard matrices of order $28$, Discrete\ Math.\ 76 (1989), 45--49.
    \bibitem{carlet2016complementary} C.\,Carlet and S.\,Guilley, Complementary dual codes for counter-measures to side-channel attacks, Adv.\ Math.\ Commun.\ 10 (2016), 131--150.
    \bibitem{dougherty2017combinatorics} S.\,T.\,Dougherty, J.-L.\,Kim, B.\,Ozkaya, L.\,Sok and P.\,Sol\'{e}, The combinatorics of LCD codes: linear programming bound and orthogonal matrices, Int.\ J.\ Inf.\ Coding\ Theory\ 4 (2017), 116--128.
    \bibitem{galvez2018somebounds} L.\,Galvez, J.-L.\,Kim, N.\,Lee, Y.\,G.\,Roe and B.-S.\,Won, Some bounds on binary LCD codes, Cryptogr.\ Commun.\ 10 (2018), 719--728.
    \bibitem{grassl-codetables} M.\,Grassl, Code tables: Bounds on the parameters of various types of codes, Available online at \url{http://www.codetables.de/}, Accessed on 5 Aug 2021.

    \bibitem{guenda2018constructions} K.\,Guenda, S.\,Jitman and T.\,A.\,Gulliver, Constructions of good entanglement-assisted quantum error correcting codes, Des.\ Codes\ Cryptogr.\ 86 (2018), 121--136.
    \bibitem{harada1996existence} M.\,Harada, Existence of new extremal doubly-even codes and extremal singly-even codes, Des.\ Codes\ Cryptogr.\ 8 (1996), 273--283.
    \bibitem{harada2006self} M.\,Harada, Self-orthogonal $3$-$(56,12,65)$ designs and extremal doubly-even self-dual codes of length $56$, Des.\ Codes\ Cryptogr.\ 38 (2006), 5--16.
    \bibitem{harada2021construction} M.\,Harada, Construction of binary LCD codes, ternary LCD codes and quaternary Hermitian LCD codes, Des.\ Codes\ Cryptogr.\ (to appear) \url{https://doi.org/10.1007/s10623-021-00916-1}, arXiv:2101.11821.
    \bibitem{harada1998classification} M.\,Harada, T.\,A.\,Gulliver and H.\,Kaneta, Classification of extremal double-circulant self-dual codes of length up to $62$, Discrete\ Math.\ 188 (1998), 127--136.
    \bibitem{harada1995new} M.\,Harada and H.\,Kimura, New extremal doubly-even $[64, 32, 12]$ codes, Des.\ Codes\ Cryptogr.\ 6 (1995), 91--96.
    \bibitem{harada2019binary} M.\,Harada and K.\,Saito, Binary linear complementary dual codes, Cryptogr.\ Commun.\ 11 (2019), 677--696.
    \bibitem{huffman2010fundamentals} W.\,C.\,Huffman and V.\,Pless, Fundamentals of error-correcting codes, Cambridge University Press (2010).
    \bibitem{kimura1994extremal} H.\,Kimura, Extremal doubly even $(56,28,12)$ codes and Hadamard matrices of order $28$, Australas.\ J. Combin.\ 10 (1994), 171--180.
    \bibitem{mallows1973anupper} C.\,L.\,Mallows and N.\,J.\,A.\,Sloane, An upper bound for self-dual codes, Inform.\ Control.\ 22 (1973), 188--200.
    \bibitem{massey1992linear} J.\,L.\,Massey, Linear codes with complementary duals, Discrete Math.\ 106/107 (1992), 337--342.
    \bibitem{yankov2014newbinary} N.\,Yankov and M.\,H.\,Lee, New binary self-dual codes of length $50$--$60$, Des.\ Codes\ Cryptogr.\ 73 (2014), 983--996.
    \bibitem{yankov2011binary} N.\,Yankov and R.\,Russeva, Binary self-dual codes of length $52$ to $60$ with an automorphism of order $7$ or $13$, IEEE Trans.\ Inform.\ Theory\ 56 (2011), 7498--7506.
    \bibitem{yorgov1987method} V.\,Yorgov, A method for constructing inequivalent self-dual codes with application to length $56$, IEEE Trans.\ Inform.\ Theory\ 33 (1987), 72--82.

\end{thebibliography}
\end{document}